\documentclass{IEEEtran}
\usepackage{amsmath}
\usepackage{amsfonts}
\usepackage{amssymb}
\usepackage{amsthm}
\usepackage{stmaryrd} 
\usepackage{tikz}

\newtheorem{theorem}{Theorem}
\newtheorem{lemma}[theorem]{Lemma}

\newtheorem{corollary}[theorem]{Corollary}
\newtheorem{proposition}[theorem]{Proposition}

\theoremstyle{definition}
\newtheorem*{defaux}{Definition}
\newenvironment{definition}%
   {\pushQED{\hfill$\blacktriangleleft$}\defaux}%
   {\popQED\enddefaux}

\newcommand\E{\mathbb O}
\newcommand\F{\mathbb F}
\newcommand\D{\mathbb D}
\newcommand\EE{\mathsf o}
\newcommand\FF{\mathsf f}
\newcommand\DD{\mathsf d}
\newcommand\tp{\ensuremath{\mathsf{tp}}}
\def\(#1,#2,#3){\langle #1,#2,#3\rangle}
\newcommand\setv{\mathbb{V}}
\newcommand\setn{\mathbb{N}}
\newcommand\parent{\mathord{\shortuparrow}}
\newcommand\FS{\Phi}
\newcommand\FSa{\Psi}
\newcommand\cbrk{\epsilon}
\newcommand\indep{\mathbin{\mkern-.3\thinmuskip\parallel\mkern-.3\thinmuskip}}
\newcommand\notindep{\mathbin{\mkern-.3\thinmuskip\nparallel\mkern-.3\thinmuskip}}
\newcommand\xlt{\sqsubset}
\newcommand\xle{\sqsubseteq}

\makeatletter
\DeclareRobustCommand
 \nll{\mathrel{\m@th\mathpalette\c@ncel\ll}}
\makeatother
\newcommand\orderrel{\ll}
\newcommand\notorderrel{\nll}
\newcommand\wequiv{\buildrel\raise-0.4ex\hbox{$\scriptstyle\mathsf{w}$}\over\equiv}
\newcommand\wxle{\mathrel{\ooalign{\hfil\raise1.4ex\hbox{$\,\scriptstyle\mathsf{w}$}\hfil\crcr
  \raise-0.4ex\hbox{$\xle$}}}}
\newcommand\wxlt{\mathrel{\ooalign{\hfil\raise1.4ex\hbox{$\,\scriptstyle\mathsf{w}$}\hfil\crcr
  \raise-0.4ex\hbox{$\xlt$}}}}
\newcommand\X{\mathcal{X}}
\let\models\vDash
\let\nmodels\nvDash
\newcommand\captp[2]{#1\cap^\tp#2}
\newcommand\mintp[2]{#1\setminus^{\!\tp}#2}
\renewcommand\setminus{\mathbin{\mskip-0.6\thinmuskip\smallsetminus\mskip-0.6\thinmuskip}}
\newcommand\R[2]{\mathcal R(#1\mathbin{|}#2)}

\newcommand{\rot}[3]{#3#2#1} %
\newenvironment{Keywords}[1]{\xdef\IEEEkeywordsname{#1}\IEEEkeywords}{\endIEEEkeywords}

\title{Algebra of Data Reconciliation}
\author{%
Elod P. Csirmaz${}^*$\thanks{${}^*$e-mail: \rot{\rot{maz.}{csir}{ep}com}{@}{elod}}
and
Laszlo Csirmaz${}^{**}$\thanks{${}^{**}$e-mail: csirmaz@renyi.hu
R\'enyi Insititute, Budapest, and UTIA, Prague.}}

\begin{document}
\maketitle

\begin{abstract}

With distributed computing and mobile applications becoming ever more
prevalent, synchronizing diverging replicas of the same data is a common
problem. Reconciliation -- bringing two replicas of the same data structure as
close as possible without overriding local changes -- is investigated in an
algebraic model. Our approach is to consider two sequences of simple
commands that describe the changes in the replicas compared to the original
structure, and then determine the maximal subsequences of each that can be
propagated to the other.
The proposed command set is shown to be functionally complete, 
and an update detection algorithm is presented which produces a command
sequence transforming the original data structure into the replica
while traversing both simultaneously.
Syntactical characterization is provided in terms of a rewriting system for
semantically equivalent command sequences. Algebraic properties of sequence
pairs that are applicable to the same data structure are investigated. Based on
these results the reconciliation problem is shown to have a unique maximal
solution. In addition, syntactical properties of the maximal solution allow
for an efficient algorithm that produces it.



\end{abstract}

\section{Introduction}\label{sec:intro}

Synchronization of diverging copies of some data stored on several
independent devices is a mechanism we nowadays take for granted. Examples are
accessing and editing calendar events, documents, spreadsheets, or
distributed internet and web services hosted in the cloud.
While there are numerous commercially available synchronization tools
for, for example, the case of filesystems \cite{KRSD,MPV,PV,TSR,TTPDSH}, 
collaborative editors \cite{OT2,OT1}, distributed databases \cite{CFRD1},
they mostly lack sound theoretical foundation.
The main aim of this paper is to provide such a mathematical framework 
for a simple, nevertheless very important special case of filesystem
synchronization. We consider the first stage of synchronization which is to
determine the maximal amount of changes which
can be safely carried over to the other copy without the need of conflict
resolution.

The synchronization paradigm follows the one described in \cite{BP} and
depicted on Figure \ref{fig:sync-process}. 
\begin{figure}[hbt]%
~~~~\begin{tikzpicture}[shorten >=2pt,scale=0.8]
\def\ss{\scriptstyle}
\node[draw] (q0) at (-1.5,0) {$\FS$};
\node[draw] (q1) at (0,2.2) {$\FS$};
\node[draw] (q2) at (0,-2.2) {$\FS$};
\node[draw] (q3) at (3.5,2.2) {$\FS_1$};
\node[draw] (q4) at (3.5,-2.2) {$\FS_2$};
\node[draw] (q5) at (7.2,2.2) {$\FS'_1$};
\node[draw] (q6) at (7.2,-2.2) {$\FS'_2$};
\node       (q24) at (1.8,-2.1) {};
\node       (q13) at (1.8,2.1)  {};
\draw (0,0) node[rotate=90]{\footnotesize \dots\dots~~replicas~~\dots\dots};
\node[draw] (u1) at (2.2,0.7) {\footnotesize update detector};
\node       (u11) at (3.3,0.7) {};
\node[draw] (u2) at (2.2,-0.7) {\footnotesize update detector};
\node       (u22) at (3.3,-0.7) {};
\node[draw,rotate=90] (R) at (5.3,0) {\footnotesize ~~reconciler~~};
\node       (ap) at (5.3,2.2) {};
\node       (bp) at (5.3,-2.2) {};
\draw (6.8,0) node [rotate=90] {\footnotesize conflicts};
\node         (cr) at (6.8,0) {};
\path[->]
      (q0) edge (q1) edge (q2)
      (q1) edge node[fill=white]{$\ss\alpha_0$} node[above=3pt]{\footnotesize updates} (q3)
      (q2) edge node[fill=white]{$\ss\beta_0$} node[below=3pt]{\footnotesize updates} (q4)
      (q3) edge node[fill=white]{$\ss\beta'$} (q5)
      (q4) edge node[fill=white]{$\ss\alpha'$} (q6);
\path[dotted,->]
      (q13) edge (u1) (q24) edge (u2);
\path[thick,->]
      (q2) edge (u2) (q4) edge (u2) 
      (q1) edge (u1) (q3) edge (u1)
      (u11) edge node[fill=white]{$\ss\alpha$} (R)
      (u22) edge node[fill=white]{$\ss\beta$} (R)
      (R) edge (ap) edge (bp) edge (cr);
\end{tikzpicture}%
\caption{The synchronization process}\label{fig:sync-process}
\end{figure}
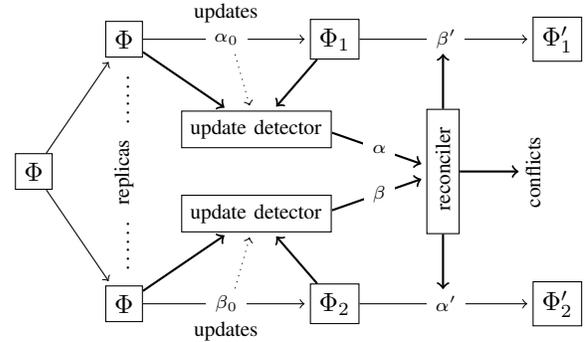
Two identical replicas of the original filesystem $\FS$ are updated independently
yielding the diverged copies $\FS_1$ and $\FS_2$. In the state-based case
the \emph{update detector} compares the original
($\FS$) and current ($\FS_1$ or $\FS_2$) states, and
extracts an update information describing the differences between the
original and the replica. 
In the operation-based case the update detector gets only the
performed operations and works on this sequence. The \emph{reconciler} 
(which represents the first stage of the full reconciliation process)
uses the information provided by the update detectors to 
propagate as many of the changes to the other replica as possible without
destroying local changes. The remaining updates are marked as \emph{conflicts}
and should be handled by a separate conflict resolver.
The details of such 
final conflict resolution is outside of the scope of this paper.
Our reconciler determines the updates which
can be safely carried over to the other replica, and pinpoints the ones
conflict and require further action.
In contrast, the goal of Operational Transformation \cite{OT1}
or the Conflict-free Replicated Data Type \cite{CFRD1} is 
full synchronization where the main focus is on conflict resolution.

In terms of the usual classification of synchronizers (see, e.g.,
\cite{PV,SSH,TSR}), our approach is an operation-based one as changes in the
filesystems are modeled as the effects of specific command
sequences. The task of the update detector is to produce such a (normalized {or canonical}
command sequence, and that of the reconciliation algorithm is to identify the 
commands which can be propagated to the other replica.
The central problem during both update detection and
reconciliation is the ordering (scheduling) of these commands. If update
detection is based on comparing the original state of the system to its
new state, then one can easily collect a set of commands that can create the
differences, but they must be ordered (if at all possible) in a way that the sequence
can be applied to the initial system without causing an error.
Similarly, during reconciliation, typically many maximal propagable
command sequences exist. For more information, we refer to the excellent
survey of the so-called optimistic replication algorithms \cite{SSH}, and,
as working examples, to IceCube, where multiple orders are tested to find an
acceptable one \cite{KRSD,MPV}, or Bayou, where reconciling updates happens
by redoing them in a globally determined order \cite{TTPDSH}.

This paper continues and extends the investigation started in \cite{epcs}.
An important contribution of \cite{epcs} is the abstract notion of
hierarchical data structure and data manipulating commands which,
on the one hand, faithfully represent real
data structures and typical commands on them, and, on the other, have a rich and
intriguing algebraic structure. This structure is investigated here
in a more general setting, while 
additional structural results and characterizations are provided. Based on these theorems
it is proved that even in the extended setting the reconciliation problem
has a unique maximal solution in a very strong sense, and that the solution can
be found using a simple and efficient algorithm.

While the exposition uses the specific
terminology of filesystems and filesystem commands,
the results can form the basis of treating the synchronization of other
structured data, like \textsf{JSON} or \textsf{XML} documents,
stored in an algebraic fashion, with similarly strong
theoretical foundations.

\smallskip

The paper is organized as follows. Hierarchical data structures and
commands on which the algebraic setting is based are defined in Section
\ref{sec:definition} as \emph{filesystems} and \emph{filesystem commands}.
Typical ``real'' filesystems and filesystem commands
can be easily modeled using these abstract notions. The suggested command
set is complete in the sense that any filesystem can be transferred into any
other one by an appropriately chosen command sequence. 
Section \ref{sec:results} summarizes the main results in an informal way.
Command pairs are
considered in Section \ref{sec:pairs} where a set of (syntactical) rewriting
rules are defined. These rules define a syntactical consequence relation
among command sequences (one sequence can be derived from the other one
using the rewriting rules), which is shown to coincide with the semantic
consequence relation for an important subset of command sequences. This
subset, called \emph{simple sequences}, is introduced in Section
\ref{sec:simple-sequences} where additional properties are stated and
proved. The update detection algorithm is detailed in Section
\ref{sec:update}. In reference to rewriting systems, sequence pairs
applicable to the same filesystem are called \emph{refluent}. Understanding
their structure requires significant effort. Section \ref{sec:refluent}
proves several partial results which are used in Section
\ref{sec:reconciliation} where the uniqueness of the maximal reconciliation
is proved together with the correctness of the algorithm identifying it.
Finally Section \ref{sec:conclusion} summarizes the results and lists some
open problems.

\section{Definitions}\label{sec:definition}

Informally, a \emph{filesystem} is modeled as a function which populates
some fixed virtual namespace $\setn$ with values from a set $\setv$. 
Elements of this virtual namespace $\setn$
are \emph{nodes}, which are arranged into a tree-like structure
reflecting the hierarchical structure of the namespace. A
filesystem assigns values to each virtual node such that along each path
starting from a root there are finitely many directories, followed by an
optional file value, followed by empty values.
This virtual namespace can actually reflect the (fully qualified) names of
all imaginable files and directories. In this case, renaming a file (or a
directory) means that the contents of the file (or the whole directory subtree) is
moved from one location to another.

\subsection{Node structure}

Formally, we fix an arbitrary and possibly infinite \emph{node structure $\setn$} endowed with the
partial function $\parent:\setn\to\setn$ which returns the parent of every
non-root node (it is not defined on roots of which there might be several). 
This function must induce a tree-like structure, which means that there
must be no loops or infinite forward chains.
We say that \emph{$n$ is above $m$} and write $n\prec m$
if $n=\parent^i(m)$ for some $i\ge 1$.
As usual, $n\preceq m$ means $n\prec m$ or $n=m$, which is a partial order on 
$\setn$. Minimal elements of this partial order are  
the roots of $\setn$. 
The nodes $n$ and $m$ are \emph{comparable} if either $n\preceq m$ or $m\preceq m$,
and they are \emph{uncomparable} or \emph{independent}, written as $n\indep
m$, otherwise.

\subsection{Filesystem values, filesystems}

A filesystem populates the nodes with values from
a set $\setv$ of possible filesystem values. $\setv$ is partitioned into
directory, file, and empty values as $\setv=\D\cup\F\cup\E$. For a value
$v\in\setv$ its \emph{type}, denoted by $\tp(v)$, is the partition it
belongs to, thus it is one of $\D$, or $\F$ or $\E$.
A \emph{filesystem} is a function $\FS:\setn\to\setv$ which has the
\emph{tree property}: along any branch starting from a root there are zero
or more directory values, then zero or one file value which is followed by
empty values only.%
\footnote{This definition allows filesystems with infinitely
many non-empty nodes, requiring only that every branch is eventually
empty. The set $\F$ reflects all possible file contents together with
additional metainformation. Also, the set $\E$ of ``empty'' values is not
required to have a single element only.}
We assume that neither $\D$ nor $\E$ is empty, and $\F$
has at least two elements.
The collection of all filesystems is denoted by $\X$.

\subsection{Filesystem commands}

The set of available filesystem commands will be denoted by $\Omega$, and the
application of a command $\sigma\in\Omega$ to the filesystem $\FS\in\X$ will be
written as the left action $\sigma\FS$. If $\sigma$ is not applicable to
$\FS$ then we say that \emph{$\sigma$ breaks $\FS$}, which is denoted by $\sigma\FS=\bot$.
Thus filesystem commands are modeled as functions mapping $\X$
into $\X\cup\{\bot\}$, where $\bot$ indicates failure.

The actual command set is specified so that, on one hand, it reflects the
usual filesystem commands, and, on the other, it is more symmetric and more
uniform. It is well-known that the actual choice of the command set has
profound impact on whether reconciliation is possible or not, see
\cite{CBNR,NREC}, and one of the main contributions of \cite{epcs} is the
systematic symmetrization of the traditional filesystem command set.
Accordingly, commands in $\Omega$ are represented by triplets specifying
\begin{itemize}\setlength\itemsep{0pt}
\item a node $n\in\setn$ on which the command acts, 
\item a precondition which specifies the type of the value 
the filesystem must have at node $n$ before executing the command, and
\item a replacement value to be stored at $n$.
\end{itemize}
Such a command is applicable if the precondition holds, and preforming the
replacement does not destroy the tree property of the filesystem.

Formally, the commands in $\Omega$ are the triplets $\sigma=\(n,t,x)$
where $n\in\setn$ is a node, $t\in\{\D,\F,\E\}$
specifies the precondition by requesting the current filesystem value at $n$
to have type $t$, and $x\in\setv$ is the replacement value. 
The \emph{input} and \emph{output type} of the
command $\sigma=\(n,t,x)$ is $t$ and $\tp(x)$, respectively.
The effect of the command $\(n,t,x)$ on the filesystem $\FS\in\X$ is
defined as
$$
  \(n,t,x)\FS = \begin{cases}
     \FS_{[n\to x]} &\mbox{if}\parbox[t]{0.56\displaywidth}{
             ~$\tp(\FS(n))=t$ (precondition)\\
             and $\FS_{[n\to x]}\in\X$ (tree property);}\\
     \bot           &\mbox{otherwise,}
\end{cases}
$$
where the operator $\FS_{[n\to x]}$ changes the value of
the function $\FS$ only at $n$ to $x$ as
$$
  \FS_{[n\to x]}(m)=\begin{cases}
    x & \mbox{if $m=n$,}\\
    \FS(m) & \mbox{otherwise.}
  \end{cases}
$$
If nodes encode the used namespace, then, for example, the typical
filesystem command ``\emph{rmdir $n$}''
corresponds to the abstract command $\(n,\D,\EE)$ where $\EE\in\E$ is some
empty value; this command fails when either $n$ is
not a directory (input type mismatch), or the directory is not empty (after
the replacement the tree property is violated). Similarly, editing a file at
location $n$ is captured by the command $\(n,\F,\FF)$ where $\FF\in\F$ represents the
new file content. 

The stipulation that illegal commands break the filesystem can be
considered as a semantical check. It is a natural requirement
that no command should leave the system in an erroneous state even 
with the promise that it will be corrected later. This requirement of ``local
consistency'' is different from the Eventual
Consistency of \cite{CFRD1}, which requires that the reconciliation process
eventually results in identical copies.

Two additional filesystem commands, denoted by $\cbrk$ and $\lambda$, will be
defined. In practice they do not naturally occur, and are not elements of
$\Omega$, but they are are useful when arguing about command sequences. The
command $\cbrk$ breaks every filesystem while $\lambda$ acts as identity:
$\cbrk\FS=\bot$ and $\lambda\FS=\FS$ for every $\FS\in\X$. These commands
have no nodes, or input or output types.
Commands in $\Omega$ are denoted by $\sigma$, $\tau$ and $\omega$.

\subsection{Command categories}

A command $\(n,t,x)$ matches the pattern $\(n,\mathsf{T},\mathsf{X})$ if
$t\subseteq\mathsf T$ and $x\in\mathsf X$. Only data types $\D$, $\F$, $\E$, 
and their unions will
be used in place of $\mathsf T$ and $\mathsf X$ with the union sign omitted.
In a pattern the symbol $\cdot$  matches any value.

Depending on their input and output types commands can be partitioned into
nine disjoint classes. \emph{Structural commands} change the type of the
stored data, while \emph{transient commands} retain it. In other words,
commands matching $\(\cdot,\E,\F\D)$, $\(\cdot,\F,\E\D)$, or $\(\cdot,\D,\E\F)$
are structural commands, while those matching $\(\cdot,\F,\F)$, $\(\cdot,\E,\E)$,
or $\(\cdot,\D,\D)$ are transient ones.

Structural commands are further split into \emph{up} and \emph{down}
commands, where up commands ``upgrade'' the type from $\E$ to $\F$ to $\D$,
while down commands ``downgrade'' the type of the stored value. That is,
up commands are those matching $\(\cdot,\E,\F\D)$ and $\(\cdot,\F,\D)$,
while down commands are matching $\(\cdot,\D,\F\E)$ and $\(\cdot,\F,\E)$.
The type of the command $\sigma=\(n,t,x)$ is $\tp(\sigma)=\(n,t,\tp(x))$,
and then $\tp(\sigma)=\tp(\sigma')$ iff $\sigma$ and $\sigma'$ have the
same node, same input type and same output type.


\subsection{Command sequences}

The free semigroup generated by $\Omega$ is $\Omega^*$; this
is the set of finite command sequences including the empty sequence $\lambda$.
Elements of $\Omega^*$ act from left to right, that is, $(\alpha\sigma)\FS =
\sigma(\alpha\FS)$ where $\sigma\in\Omega$ and $\alpha\in\Omega^*$. This
definition is in full agreement with the definition of the special command
$\lambda$. A command sequence $\alpha$ breaks a filesystem $\FS$
if some initial segment of $\alpha$ breaks it. In other 
words, $(\sigma\alpha)\FS=\bot$ if either $\alpha\FS=\bot$ or
$\sigma(\alpha\FS)=\bot$.
We use $\alpha$, $\beta$ and $\gamma$ to denote command sequences.

We write $\alpha\xle\beta$ to denote that \emph{$\beta$ semantically extends
$\alpha$}, that is, $\alpha\FS=\beta\FS$ for all filesystems $\FS$ that
$\alpha$ does not break. Similarly,
$\alpha\equiv\beta$ denotes that $\alpha$ and $\beta$ are \emph{semantically
equivalent}, meaning $\alpha\FS=\beta\FS$ for all $\FS\in\X$. Clearly,
$\alpha\equiv\beta$ if and only if both $\alpha\xle\beta$ and
$\beta\xle\alpha$.
As $\cbrk$ breaks every filesystem, $\alpha\not\equiv\cbrk$ means that
$\alpha$ is defined on some filesystem.
In this case we say that $\alpha$ is a \emph{non-breaking sequence}.

For $\Delta\subseteq\Omega^*$ and $\alpha\in\Omega^*$, $\Delta\models\alpha$
denotes that for every filesystem $\FS$, if none of $\delta\in\Delta$ breaks
$\FS$, then neither does $\alpha$. This relation shares many properties of
the ``consequence'' relation used in mathematical logic.
As usual, $\Delta\models\Delta'$ means that $\Delta\models\delta$ for all
$\delta\in\Delta'$, and we
also write $\delta\models\alpha$ instead of $\{\delta\}\models\alpha$.
The following claims are immediate from the definitions.
\begin{proposition}\label{prop:models}
{\upshape(a)} If $\alpha\in\Delta$ then $\Delta\models\alpha$.
{\upshape(b)} If $\Delta\models\alpha$ and $\Delta\subseteq\Delta'$, then
$\Delta'\models\alpha$.
{\upshape(c)} If $\Delta\models\Delta'$ and $\Delta'\models\alpha$ then
$\Delta\models\alpha$.
{\upshape(d)} $\Delta\models\alpha$ if and only if $\Delta'\models\alpha$ for some
finite $\Delta'\subseteq\Delta$.
{\upshape(e)} $\alpha\beta\models\alpha$.
{\upshape(f)} If $\Delta\models\alpha$, then $\{\gamma\delta:\delta\in\Delta\}\models
\gamma\alpha$.
\qed
\end{proposition}
Property (f) is an analog of the preconditioning property in logical
systems.
Observe that $\alpha$ is non-breaking iff $\alpha\nmodels\cbrk$; and
$\delta\xle\alpha$ implies $\delta\models\alpha$ as the latter only requires
that $\alpha$ is defined where $\delta$ is defined, while the former also
requires that where they are both defined their effect is the same.

For two sequences $\{\alpha,\beta\}\nmodels\cbrk$ iff there is a filesystem
on which both $\alpha$ and $\beta$ are defined.
With an eye on rewriting systems \cite{rewriting}, such a pair is called
\emph{refluent}.


\section{The results}\label{sec:results}

\subsection{Update detection and reconciliation}

A command-based reconciliation system works with two command sequences $\alpha$ and
$\beta$ that have been applied to a single filesystem $\FS$ yielding two different
replicas $\FS_1$ and $\FS_2$ which we need to reconcile. While it is
conceivable that the sequences are based on records of the executed
filesystem operations (operation-based reconciler),
in several filesystem implementations no such records
exist. In these cases the command sequences must be created by comparing
$\FS_1$ (or $\FS_2$) to $\FS$ (state-based reconciler). This process is called \emph{update
detection,} in which we also include
transforming the resulting (or provided) command sequence into
a canonical form required by reconciliation.
The term \emph{canonical sequence} is used informally for sequences 
amenable for the reconciliation process. Our first result is that
the chosen set of filesystem commands has the required expressive power:
one can always find a canonical command sequence that transforms the original filesystem
into the replica.

\begin{theorem}[Informal, update detection]\label{thm:informal-1}
{\upshape(a)} Given arbitrary filesystems $\FS_1$ and $\FS$, there exists a canonical
command sequence $\alpha$ such that $\FS_1 = \alpha\FS$. The commands in
$\alpha$ can be found by traversing $\FS$ and $\FS_1$ while searching for
different node content; the order of the commands can be found in quadratic time
in the size of $\alpha$.\\
{\upshape(b)} Given any command sequence that transforms $\FS$ to $\FS_1$,
the corresponding canonical sequence can be created from it in quadratic time.
\qed
\end{theorem}

\emph{Reconciliation} is the process of merging the diverging replicas
as much as possible, and mark cases where it is not possible to do so
without further\allowbreak---usually human---\allowbreak input as \emph{conflicts.}
In our approach it means applying as many updates (commands) as possible that have been
applied to one replica to the other
without breaking the filesystem or overriding local changes.
The remaining commands are marked as conflicting updates.
Resolving these conflicts, as they require knowledge and input not available
in the filesystems, is outside the scope of the reconciliation algorithm.
More formally, $\beta'$ is a \emph{reconciler for $\alpha$ over
$\beta$} if 
\begin{itemize}\setlength\itemsep{0pt}
\item $\beta'$ consists of commands from $\beta$;
\item for any $\FS$, $\beta'$ is applicable to $\alpha\FS$ whenever both
$\alpha\FS$ and $\beta\FS$ are defined;
\item no command in $\beta'$ overrides the effect of any command from
$\alpha$.
\end{itemize}
The main result of this paper is that in this algebraic framework the 
reconciliation problem has a unique maximal solution in a very strong sense.
\begin{theorem}[Informal, reconciliation]\label{thm:informal-2}
Suppose $\alpha$ and $\beta$ are canonical sequences, and there is at least
one filesystem neither of them breaks. Then there is a maximal reconciler
$\beta'$ for $\alpha$ over $\beta$
which can be created from $\alpha$ and $\beta$ in quadratic time.
Moreover, $\beta'$ is optimal in a very strong sense:
for any sequence $\beta''$ consisting of commands of $\beta$, if
$\beta''$ contains a command not in $\beta'$, then either $\beta''$
overrides a change made by $\alpha$, or $\alpha\beta''$ breaks every
filesystem.
\qed
\end{theorem}
Theorem \ref{thm:informal-1} is proved in Section \ref{sec:update} as
Theorems \ref{thm:update-1} and \ref{thm:update-2}, while Theorem
\ref{thm:informal-2} follows from Theorem \ref{thm:reconc} in Section
\ref{sec:reconciliation} and the discussion following its proof.


\subsection{Limitations and extensions}

Our model of filesystems is intentionally simple, but it turns out it is
this property that can guarantee locality and allow the highly symmetric set
of filesystem commands considered in this paper to have the necessary
expressive power. For example, the filesystem model does not consider
metadata on the nodes including timestamps or permissions. While these could
be modeled as contents of special nodes under the one they relate to, such
a model would force reconciliation to be considerably more complicated as
there would be an extra dependency between these nodes. As the parent node
(a file or directory) cannot be created or modified without also setting its
metadata, a conflict on the metadata would need to be propagated to the
parent. If metadata is to be included in the current model, it is likely to
be more fruitful to consider it part of the contents of a file or directory;
however, this would add extra burden on the conflict resolver.

Hard and symbolic links, as well as a possible \emph{move} or \emph{rename} 
command, pose
harder problems as they destroy locality. Our proofs extensively use the
fact that commands on unrelated nodes (no one is an ancestor of the other)
commute, and in the presence of links this property ceases to hold. Possible
extensions of the model that would handle links could be based on so-called
\emph{inodes}, where contents in the main file system are merely pointers to the
real file contents stored elsewhere, allowing multiple nodes to reference
the same content. Commands modifying the pointers and the set of contents
would be separated and investigated separately for conflicts.

A \emph{move} or \emph{rename} command, which, unlike the ones in our model, affects the
filesystem at two nodes at once, can prove convenient as it is easy for
human reviewers to understand and verify that no data is lost, and can move
data with minimal overhead. It could be introduced to the synchronizer system
by breaking it up into a \emph{delete} and a \emph{create} command before it is considered
by the algorithms defined in this paper. The outputs of these algorithms,
the reconciler and conflicting command sequences, could be post-processed in
turn to re-introduce \emph{move} commands by merging a \emph{delete} and a
\emph{create} command wherever possible.

Despite these limitations, our results can immediately be carried over to
any structured data that can be transformed into a tree-like structure, for
example, data in \textsf{JSON} or \textsf{XML} formats. 
See also the \emph{lens} concept in \cite{FGKPS}.

Finally, Theorem \ref{thm:informal-2} can be used as an ``advisor'' for the
final conflict resolver. Conflicting commands either modify the same node in a different way
(first case) requiring content negotiation or other priority consideration,
or one of the commands or command sequences must be rolled back (second
case). Having resolving a conflict, the Theorem can be applied repeatedly
until all conflicts are resolved.


\medskip

\section{Rules on command pairs}\label{sec:pairs}

\begin{proposition}[Command pairs 1]\label{command-pairs-1}
Suppose $\sigma,\tau\in\Omega$ are on the same node $n$. Then exactly one of the
following possibilities hold:\\
{\upshape(a)} $\sigma\tau\xle\omega$ for some $\omega\in\Omega$ also on node $n$,\\
{\upshape(b)} $\sigma\tau\equiv\cbrk$.
\end{proposition}

\begin{proof}
Let the two commands be $\sigma=\(n,t,x)$ and $\tau=\(n,q,y)$, respectively.
If $\tp(x)\neq q$, then case (ii) holds. If
$\tp(x)=q$, then the combined effect of the commands is the same as
that of the command $\omega=\allowbreak\(n,t,y)$. In general only $\xle$ is 
true as $\sigma=\(n,t,x)$ could break a filesystem on which 
$\omega=\(n,t,y)$ works.
\end{proof}

To maintain the tree property, certain command pairs on successive nodes
can only be executed in a
certain order. This notion is captured by the binary relation
$\sigma\orderrel\tau$.

\begin{definition}
For a command pair $\sigma$, $\tau\in\Omega$ 
the binary relation $\sigma\orderrel\tau$ holds if the pair matches 
either $\(n,\D\F,\E)\orderrel\(\parent n,\D,\F\E)$ or $\(\parent
n,\E\F,\D) \orderrel \(n,\E,\F\D)$.
\end{definition}
Observe that $\sigma\orderrel\tau$ implies that
$\sigma$ and $\tau$ are structural commands on consecutive nodes,
and either both are up commands, or both are down commands.
Also, if $\sigma_1\orderrel\sigma_2\orderrel\sigma_3$ then all three
commands are in the same category, thus the corresponding nodes are
going up or going down. In particular, there are no commands which
would form a $\orderrel$-cycle.

\begin{definition}
Let $\sigma$ be on node $n$, and $\tau$ be on a different node $m$. The
(symmetric) binary relation $\sigma\indep\tau$ holds in the following cases:
$n$ and $m$ are uncomparable; or if $n$ and $m$ are comparable then
either the command on the higher node matches $\(\cdot,\D,\D)$, or the
command on the lower node matches $\(\cdot,\E,\E)$, or both.
\end{definition}

\begin{proposition}[Command pairs 2]\label{command-pairs-2}
Suppose $\sigma$ and $\tau$ are on different nodes.\\
{\upshape(a)} $\sigma\indep\tau$ if and only if $\sigma\tau \equiv\tau
\sigma \not\equiv\cbrk$;
\\
{\upshape(b)} if $\sigma\notindep\tau$ then $\sigma\tau\not\equiv\cbrk
\Leftrightarrow \sigma\orderrel
\tau$.
\end{proposition}

\begin{proof}
Tedious, but straightforward case by case checking.
\end{proof}

An immediate consequence of Proposition \ref{command-pairs-2} is

\begin{proposition}\label{remark-R1}
If $\sigma$ and $\tau$ are on different nodes, then exactly one of the
following three possibilities hold: $\sigma\orderrel\tau$, or
$\tau\orderrel\sigma$, or $\sigma\tau\equiv\tau\sigma$. \qed
\end{proposition}


\subsection{Rewriting rules}

Statements in Propositions \ref{command-pairs-1} and \ref{command-pairs-2}
can be considered as \emph{rewriting rules} on command sequences where the
command pair $\sigma\tau$ on the left hand side can be replaced by
one or two commands on the right hand side. Let us summarize these rewriting
rules for future use.

\begin{proposition}[Rewriting rules]\label{prop:rules}
For a command pair $\sigma\tau$,\\
{\upshape(a)} if $\sigma$, $\tau$ are on the same node, then either
$\sigma\tau\equiv\cbrk$, or 
$\sigma\tau\xle\omega$ for some $\omega\in\Omega$
which is on the same node as $\sigma$ and $\tau$ are;\\
{\upshape(b)} if $\sigma$, $\tau$ are on different nodes and
$\sigma\notorderrel\tau$, then $\sigma\tau\equiv\tau\sigma$ if
$\sigma\indep\tau$, and 
$\sigma\tau\equiv\cbrk$ otherwise.
\qed
\end{proposition}

The syntactical rewriting rules indicated in Proposition \ref{prop:rules} fall into three patterns:
\begin{itemize}
\item[] $\sigma\tau\equiv\tau\sigma$ ($\sigma$ and $\tau$ commute);
\item[]
$\sigma\tau\equiv\cbrk$ (the pair, in this order, breaks every
filesystem);
\item[] $\sigma\tau\xle\omega$ for some single command $\omega$.
\end{itemize}
In the first two cases the rule preserves semantics, while in the
last case extends it.
To handle breaking sequences seamlessly, three 
additional rules are added expressing that $\cbrk$ is an absorbing 
element \cite{absorbing}:
$$
   \cbrk\cbrk\equiv\cbrk, \quad
   \sigma\cbrk\equiv\cbrk, \quad
   \cbrk\sigma\equiv\cbrk.
$$

\begin{definition}
For two command sequences 
$\alpha\wxle\beta$ denotes that there is a rewriting sequence 
using the above rules 
which produces $\beta$ from $\alpha$; and $\alpha\wequiv\beta$ denotes that
there is a rewriting sequence using semantic preserving rules only.
\end{definition}

Observe that $\alpha\wequiv\beta$ is not symmetric but
clearly transitive.
With an abuse of notation, we write
$\alpha\wequiv\beta$ to mean that either both $\alpha\wequiv\cbrk$ and
$\beta\wequiv\cbrk$, or $\alpha\wequiv\beta$ 
which clearly makes $\wequiv$ symmetric.
This extended notation lets us rephrase Proposition \ref{remark-R1} in terms of rewriting:

\begin{proposition}\label{remark-R1*} 
If $\sigma$ and $\tau$ are on different nodes, then exactly one of the
following three possibilities hold: $\sigma\orderrel\tau$,
$\tau\orderrel\sigma$, or $\sigma\tau\wequiv\tau\sigma$.\qed
\end{proposition}


\subsection{Functional completeness}

The command set $\Omega$ is sufficiently rich to allow transforming any
filesystem into any other one assuming that they differ at finitely many nodes
only. The proof is constructive, meaning that it not only proves the existence of,
but actually specifies an algorithm that creates, such a sequence.

\begin{theorem}\label{thm:sem-compl}
The command set $\Omega$ is complete in the following sense. Let $\FS_0$ and
$\FS_1$ be two filesystems differing at finitely many nodes only. There
is a command sequence $\alpha\in\Omega^*$ which transforms the first
filesystem to the other one as $\alpha\FS_0=\FS_1$.
\end{theorem}

\begin{proof}
By induction on the number of nodes $\FS_0$ and $\FS_1$ differ. If this
number is zero, let $\alpha$ be $\lambda$. Otherwise let $n$ be 
one of the lowest
nodes where $\FS_0(n)=x_0$ and $\FS_1(n)=x_1$ differ, that is, where $\FS_0(m)=
\FS_1(m)$ for every node $m$ below $n$.

Let $\FS'_0$ be $\FS_{0[n\to x_1]}$
and $\FS'_1$ be $\FS_{1[n\to x_0]}$.
Clearly, the number of nodes at which $\FS'_0$ and $\FS_1$ 
differ is one less.
If $\FS'_0$ is not broken, then the induction gives $\alpha'$ for which $\alpha'\FS'_0=\FS_1$, 
and we set $\alpha$ to $\(n,\tp(x_0),x_1)\alpha'$. This does not break the
filesystem because $\FS'_0=\(n,\tp(x_0),x_1)\FS_0$. 
Similarly, if $\FS'_1$ is not broken, then the induction gives $\alpha'\FS_0=\FS'_1$
for some $\alpha'$, and we set $\alpha$ to $\alpha'\(n,\tp(x_0),x_1)$.

It remains to show that either $\FS'_0\in\X$ or $\FS'_1\in\X$, which holds 
if the corresponding function has the tree property. It is trivial if $x_0$ and $x_1$ have the same data
type. Otherwise, as $\FS_0$ and $\FS_1$ have the same values below $n$, 
the filesystem in which the value at $n$ is downgraded 
(in the sense that $\D>\F>\E$)
will retain the tree property.
\end{proof}

Actually, a stronger statement has been proved. The command sequence
transforming $\FS_0$ to $\FS_1$ consists of the commands
$$
  \{  \(n,\tp(\FS_0(n)),\FS_1(n)) \,: n\in\setn \mbox{ and } \FS_0(n)\neq\FS_1(n) \}
$$
in some order. In particular, each command in $\alpha$ is on a different node.


\section{Simple sequences}\label{sec:simple-sequences}

Since rewriting rules are semantically correct, $\alpha\wequiv\beta$ implies
$\alpha\equiv\beta$, and $\alpha\wxle\beta$ implies $\alpha\xle\beta$. The
natural question arises whether this set of rewriting rules is
\emph{complete}, meaning that the converse implication also holds:
$\alpha\equiv\beta$ implies $\alpha\wequiv\beta$. We will prove in Theorem
\ref{thm:simple-completeness} that this is indeed the case for the special
class of \emph{simple sequences}.

\begin{definition}
A finite command sequence is \emph{simple} if it contains at most one command on
each node.
\end{definition}

For example, the command sequence $\(n,\E,\FF)\cdots\(n,\F,\FF')$ creates a file at node $n$ 
and then later edits it. It is \emph{not} simple as it touches the node $n$
twice. The sequence $\(n,\D,\EE)\cdots\(n,\E,\FF)$ deletes the directory at $n$,
and then later creates a file there; it is also not simple. The following sequence
deletes the three files under a directory at $n$, and then deletes the directory itself:
\begin{equation}\label{eq:1}
   \(n_1,\F,\EE)\(n_2,\F,\EE)\(n_3,\F,\EE)\(n,\D,\EE).
\end{equation}
This sequence is simple if the nodes $n_i$ are
different and $\parent n_i=n$.

Simple sequences form a semantically rich subset of $\Omega^*$ as
every command sequence can be turned into a simple one while extending its
semantics (Theorem \ref{thm:rewriting}). At the same time non-breaking
simple sequences have strong structural properties (Theorem
\ref{thm:struct}), which makes them suitable for reconciliation.

\begin{definition}
(a) The simple sequence $\alpha\in\Omega^*$ \emph{honors $\orderrel$} provided
that if two
commands $\sigma$ and $\tau$ from $\alpha$ satisfy $\sigma\orderrel\tau$,
then $\sigma$ precedes $\tau$ in $\alpha$.

\noindent (b)
The command $\sigma\in\alpha$ is
a \emph{leader in $\alpha$} if it is $\orderrel$-minimal, that is, there is
no $\tau\in\alpha$ for which $\tau\orderrel\sigma$. In particular, each transient
command is a leader.
\end{definition}

As an example, in sequence (\ref{eq:1}) the commands
$\(n_i,\F,\EE)$ are leaders, while the last command is not. In the sequence
$$
   \(\parent n,\E,\DD) \(n,\E,\DD') \(n_1,\E,\FF_1)\(n_2,\E,\FF_2)\(n_3,\E,\FF_3)
$$
creating a directory $\DD$ at $\parent n$, a directory $\DD'$ under $\DD$ and three files
under $\DD'$, the only leader is the first command.

\begin{lemma}\label{lemma-7}
{\upshape(a)} Suppose for a simple sequence $\alpha=\alpha_1\sigma\alpha_2$ where $\sigma$ is a
leader in $\alpha$. If $\alpha\not\wequiv\cbrk$, then $\alpha\wequiv\sigma\alpha_1\alpha_2$.\\
{\upshape(b)} If the simple sequence $\alpha$ does not honor $\orderrel$, then
$\alpha\wequiv\cbrk$.\\
{\upshape(c)} Let $\beta$ be a permutation of the simple sequence $\alpha\not\wequiv\cbrk$ such that
$\beta$ honors $\orderrel$. Then $\alpha\wequiv\beta$.
\end{lemma}

\begin{proof}
All statements are consequences of the facts that exactly one of
$\sigma\tau\wequiv\tau\sigma$, $\sigma\orderrel\tau$ and
$\tau\orderrel\sigma$ holds (Proposition \ref{remark-R1*}), together with
$\tau\sigma\wequiv\cbrk$ when $\sigma\notindep\tau$ and
$\tau\notorderrel\sigma$ (Proposition \ref{command-pairs-2}).

(a) Let $\tau$ be the last command in $\alpha_1$. Then $\tau\sigma\wequiv\sigma\tau$ as we cannot
have neither $\tau\orderrel\sigma$ (as $\sigma$ is a leader), nor
$\sigma\orderrel\tau$ (as $\alpha\not\wequiv\cbrk$).

(b) Let $\sigma\orderrel\tau$ and consider the simple sequence
$\tau\alpha\sigma$. Let the last command in $\alpha$ be $\sigma_1$. If it commutes with $\sigma$, then swap
them, and continue investigating the remaining commands between $\tau$ and
$\sigma$. If they do not, $\sigma_1\orderrel\sigma$ must hold as otherwise
$\tau\alpha\sigma$ would rewrite to $\cbrk$. Then, if $\sigma_1$ commutes with
$\sigma_2$, the
command before it, swap them as before, otherwise $\sigma_2\orderrel\sigma_1$.
Ultimately the command we get following $\tau$ is $\sigma_k$ for some $k$. Now
$\tau$ and $\sigma_k$ must be on comparable nodes; both of them are structural
commands (thus $\tau\notindep\sigma_k$), and $\tau\notorderrel \sigma_k$
as there are no $\orderrel$-cycles, therefore $\tau\sigma_k\wequiv\cbrk$.

(c) By induction on the length of $\alpha$. Let $\sigma$ be a leader in
$\alpha$, then $\alpha\wequiv\sigma\alpha_1$ by (a). Let
$\beta=\beta_1\sigma\beta_2$. We claim $\tau\sigma\wequiv\sigma\tau$ for all
$\tau\in\beta_1$. It is so as $\tau\notorderrel\sigma$ ($\sigma$ is a
leader), and $\sigma\notorderrel\tau$ (since $\beta$ honors $\orderrel$). Thus
$\beta\wequiv\sigma\beta_1\beta_2$, and the induction on $\alpha_1$ and
$\beta_1\beta_2$ gives the claim.
\end{proof}

\begin{lemma}\label{lemma-1}
Let $\alpha\not\wequiv\cbrk$ be a simple sequence which contains
the commands $\sigma$ and $\tau$ on (comparable) nodes $n$ and $m$ such that
$\sigma\notindep\tau$.
Then $\sigma$ and $\tau$ are structural commands, 
and $\alpha$ contains structural commands on all nodes between $n$ and $m$.
\end{lemma}

\begin{proof}
If $n$ and $m$ are immediately related, the claim is a consequence of
Proposition \ref{remark-R1*}.

If not, let $\alpha$ be the shortest counterexample to this claim.
Then $\alpha$ must
start with $\sigma$ and end with $\tau$ (we can assume
this order), as otherwise a shorter counterexample would exist.
Also, $\alpha$ must contain more than two commands, as
by Proposition \ref{command-pairs-2}
it would otherwise rewrite to $\cbrk$.

Of $\sigma$ and $\tau$, we consider the command that is on the lower node.
If it is $\sigma$, isolate the first two commands in
$\alpha=\sigma\sigma'\beta$.
If $\sigma\sigma'\wequiv\sigma'\sigma$, then $\sigma\beta$ would be a
shorter counterexample.
Consequently, by Proposition
\ref{remark-R1*}, $\sigma\orderrel\sigma'$. It means that $\sigma$ is a
structural command, and $\sigma'$ is a structural
command on an immediate relative of $n$ which is still below $m$.
Therefore $\sigma'\notindep\tau$, and $\sigma'\beta$ would be
be a shorter counterexample. 
If $\tau$ is on the lower node, we isolate the last two commands in
$\alpha=\beta\tau'\tau$ and proceed in a similar fashion.
\end{proof}

\begin{theorem}[Rewriting theorem]\label{thm:rewriting}
For each command sequence $\alpha$ either $\alpha\wequiv\cbrk$ or
there is a simple sequence $\alpha^*$ such that $\alpha\wxle\alpha^*$.
\end{theorem}

\begin{proof}
We assume $\alpha$ is not simple.
Let $\sigma$ be the first command in $\alpha$ for which there is an earlier
command $\tau$ on the same node. Let this node be $n$.
Splitting $\alpha$ around these
commands we get
$$
     \alpha=\beta \tau\gamma\sigma\beta'.
$$
Now $\tau\gamma$ and $\gamma\sigma$ are simple sequences.
We claim that $\tau\gamma\sigma$ simplifies (or rewrites to
$\cbrk$), which is proved by induction on the length of $\gamma$. 
If $\gamma$ is empty, then it is guaranteed by Proposition
\ref{command-pairs-1}. Otherwise let $\tau'$ be the first command in $\gamma$. If
$\tau\tau'\wequiv\tau'\tau$, then the induction hypothesis gives the claim.
Thus we must have $\tau\orderrel\tau'$ (as otherwise $\tau\tau'$
would rewrite to $\cbrk$). Similarly, if the last command in $\gamma$ is $\sigma'$,
then either $\sigma'\sigma\wequiv\sigma\sigma'$, when we are done, or
$\sigma'\orderrel\sigma$.

If the length of $\gamma$ is at least two, then $\tau'$ and $\sigma'$ are on
different nodes (as $\gamma$ is simple). We also know they are on nodes that
are immediately related to $n$, and consequently $n$ is between them.
As both $\tau'$ and $\sigma'$ are structural commands, from Lemma
\ref{lemma-1} we know that $\gamma$ contains a command on $n$, which is
impossible.

Finally, if $\gamma$ has a length of one, then $\tau\orderrel\tau'
=\sigma'\orderrel\sigma$, contradicting the assumption that $\tau$ and $\sigma$
are on the same node.
\end{proof}

\begin{definition}
(a) A \emph{$\orderrel$-chain} is a command sequence $\sigma_1\orderrel \sigma_2
\orderrel \cdots\orderrel \sigma_k$ connecting $\sigma_1$ and $\sigma_k$
(or $\sigma_k$ and $\sigma_1$).

\noindent (b)
A finite set $T$ of nodes is a \emph{subtree rooted at $n\in T$} if every
other element of $T$ is below $n$, and if $t\in T$, then nodes between $t$
and $n$ are also in $T$.

\noindent (c)
Finally, $S\subset\Omega$ is a \emph{simple set} if all commands in
$S$ are on different nodes, and for any two $\sigma,\tau\in S$ either
$\sigma\indep\tau$, or $S$ contains a $\orderrel$-chain connecting
$\sigma$ and $\tau$.
\end{definition}

Some structural properties of simple sets are summarized below 
to paint an intuitive picture of their structure.
To this end
let $S\subset\Omega$ be a fixed simple set. Split the
set of nodes of the commands in $S$ into three disjoint parts $N_\D\cup
N_\E \cup N_*$ as follows. If a command in $S$ matches $\(n,\D,\D)$, then put
its node $n$ into $N_\D$; if it matches $\(n,\E,\E)$, then put $n$ into $N_\E$, otherwise
put it into $N_*$. The following statements are immediate from the
definition.

\begin{proposition}
{\upshape(a)} No node in $N_\D$ is below any node in $N_\E\cup N_*$.\\
{\upshape(b)} No node in $N_\E$ is above any node in $N_\D\cup N_*$.\\
{\upshape(c)} $N_*$ is a disjoint union of subtrees whose roots are pairwise
uncomparable.\\
{\upshape(d)} If $\sigma,\tau\in S$, $\sigma\orderrel\tau$, then $\sigma$ and $\tau$
are on consecutive nodes in the same subtree of $N_*$.
Conversely, if $T\subset N_*$ is one of
the subtrees and $\sigma$, $\tau\in S$ are commands on consecutive nodes
of $T$, then $\sigma\orderrel\tau$ or $\tau\orderrel\sigma$.
Moreover, the commands with
nodes in $T$ are either all up commands, or all down commands.\\
{\upshape(e)} Leaders ($\orderrel$-minimal elements) of $S$ are the commands
on nodes in $N_\D\cup N_\E$, on the root nodes of up-subtrees, and on the leaves of down-subtrees.
\qed
\end{proposition}

\begin{theorem}[Structural theorem of simple sequences]\label{thm:struct}
Let $\alpha\not\wequiv\cbrk$ be a simple sequence. Then\\
{\upshape(a)} $\alpha$ honors the relation $\orderrel$;\\
{\upshape(b)} if $\beta$ is a permutation of $\alpha$ and $\beta$
honors $\orderrel$, then $\beta\wequiv\alpha$;\\
{\upshape(c)} the set of commands in $\alpha$ is a simple set.\\
{\upshape(d)} Suppose the commands of the sequence $\beta$ form a simple set, and
$\beta$ honors $\orderrel$. Then $\beta\not\wequiv\cbrk$.
\end{theorem}

\begin{proof}
(a) and (b) has been proved in Lemma \ref{lemma-7}.

(c) is a direct consequence of Lemma \ref{lemma-1}. We know the commands in
$\alpha$ apply to different nodes, and that for any $\sigma\notindep\tau$
there is a chain of commands on immediately related nodes leading from one
to the other. The details of the proof show that they must also form a
$\orderrel$-chain.

(d) All applicable rewriting rules are of the form $\sigma\tau\wequiv\tau \sigma$ and
$\sigma\tau\wequiv\cbrk$. It means that $\beta\wequiv\cbrk$ if and only if
one can use the commutativity rules to rearrange $\beta$ to contain two
consecutive commands to which the $\sigma\tau\wequiv\cbrk$ rule would apply. 
Assume this is the case.
Then
$\sigma\notindep\tau$, thus $\sigma$ and $\tau$ are connected by some $\orderrel$-chain
$\sigma=\sigma_1\orderrel\cdots\orderrel\sigma_k=\tau$. We also know $k>2$. Observe that $\sigma_1$, \dots,
$\sigma_k$ must appear in $\beta$ in this order, and this order remains after
applying any commutativity rule. This is a contradiction as then $\sigma$
and $\tau$ can never become consecutive commands.
\end{proof}

\begin{definition}
For sequences $\alpha$, $\beta$ we write $\alpha\indep\tau$ to mean
$\sigma\indep\tau$ for all $\sigma\in\alpha$; and write $\alpha\indep\beta$
to mean $\alpha\indep\tau$ for all $\tau\in\beta$.
\end{definition}
An immediate consequence
of Theorem \ref{thm:struct} is the following.

\begin{proposition}\label{claim:add-one}
Let $\alpha\tau$ be a simple sequence where $\alpha\not\wequiv\cbrk$. Then\\
{\upshape(a)} $\alpha\tau\not\wequiv\cbrk$ if and only if either 
$\alpha\indep\tau$ or $\sigma\orderrel\tau$ for some $\sigma\in\alpha$.\\
{\upshape(b)} $\tau\alpha\not\wequiv\cbrk$ if and only if either
$\tau\indep\alpha$ or $\tau\orderrel\sigma$ for some $\sigma\in\alpha$.
\qed
\end{proposition}

Our next goal is to show that on simple sequences the set of
rewriting rules is semantically complete in a strong sense. To this end
we first state a result which shows that the filesystem commands capture a
surprising amount of information. If the simple sequence $\alpha$ does not
break $\FS$, then clearly $\FS$ must match the input type of every command
in $\alpha$. This simple necessary condition is almost sufficient.

\begin{theorem}\label{non-breaking}
Let $\alpha\not\wequiv\cbrk$ be a simple sequence and $\FS\in\X$ be a
filesystem. $\alpha$ does not break $\FS$ if and only if the following
conditions hold for each $\sigma\in\alpha$:\\
{\upshape(a)} If $\sigma$ is on node $n$, then $\FS(n)$ has the data type
required by $\sigma$.\\
{\upshape(b)} If $\sigma$ is a leader matching $\(n,\E,\F\D)$, then the nodes
above $n$ are directories.\\
{\upshape(c)} If $\sigma$ matches $\(n,\D,\F\E)$, then nodes below $n$ not
mentioned in $\alpha$ are empty.
\end{theorem}

\begin{proof}
The conditions are clearly necessary. For the converse use induction on the
length of $\alpha$. Let $\FS$ be a filesystem satisfying the conditions
for commands in $\tau\alpha$, where $\tau$ is on node $m$. Clearly, $\tau$
can be applied to $\FS$ as $\tau$ is a leader, thus it suffices to check
that $\tau\FS$ satisfies the conditions for $\sigma\in \alpha$. Let
$\sigma$ be on node $n$. (a) clearly holds as $m$ and $n$ are different.
For (b) observe that by Proposition \ref{claim:add-one} if $\sigma$ is a
leader in $\alpha$, then either it is a leader in $\tau\alpha$ (and then
$\tau\indep \sigma$) or $\tau\orderrel \sigma$. In the first case either
$m$ and $n$ are incomparable, or $m$ is below $n$, or $\tau$ matches
$\(n,\D,\D)$. In all cases $\tau\FS$ and $\FS$
have the same types of values above $n$.
We know $\sigma$ is an up command, so
in the second case $\tau$ is one, too, and it is on the parent node
of that of $\sigma$. $\tau\FS[m]$ is
a directory, and thus every other node above it is a directory as well.

The reason why (c) holds is similar. If $\tau\indep\sigma$, then either $m$
is above $n$, or below it and $\tau$ matches $\(m,\E,\E)$. Otherwise there is a
$\orderrel$-chain connecting $\tau$ and $\sigma$ consisting of down commands
only. Thus $\tau$ matches $\(m,\D\F,\E)$, and therefore
$\tau\FS[m]$ is empty.
\end{proof}

The following corollary, which merges the last two conditions,
will be used when constructing non-breaking filesystems.
\begin{corollary}\label{file-comparable}
Let $\alpha$ be a simple sequence, and let $\FS_1$ and $\FS_2$ be filesystems
so that $\FS_1(n)$ and $\FS_2(n)$ have the same types for every node $n$
which is {\upshape(a)} the node of some command in $\alpha$, and
{\upshape(b)} comparable to the node of a structural command in $\alpha$.
Then $\alpha\FS_1\neq\bot$ iff $\alpha\FS_2\neq\bot$.
\qed
\end{corollary}

\begin{theorem}[Completeness theorem for simple sequences]\label{thm:simple-completeness}
For a simple sequence $\alpha$, if
$\alpha\equiv\cbrk$, then $\alpha\wequiv\cbrk$. In general, if
$\alpha$ and $\beta$ are simple sequences such that $\alpha\equiv\beta$,
then $\alpha\wequiv\beta$ provided there are at least two different values
in each data type.
\end{theorem}

\begin{proof}
a) Assume $\alpha\not\wequiv\cbrk$. It suffices to
construct a filesystem $\FS$ which satisfies the conditions of
Theorem \ref{non-breaking}. Start with the empty filesystem. For each leader
$\sigma\in\alpha$ on a node $n$ that does not match $\(n,\E,\E)$, change the nodes
above $n$ in $\FS$ to a directory value, and set $n$ to a value matching the input
type of $\sigma$.

After the process finishes, condition (a) of Theorem \ref{non-breaking}
clearly holds for leaders. If $\tau$ is not a leader, then it is either an up
command or a down command. In the first case its leader is above $\tau$,
thus the node of $\tau$ is empty in $\FS$ (leaders are on uncompatible
nodes), as required. If $\tau$ is a down
command, then its leader is below $\tau$, and then the node of $\tau$ is a
directory node. Condition (b) is clear from the construction. For (c) remark
that non-empty values are only above nodes mentioned in $\alpha$.

\smallskip

b)
Assume $\alpha$ and $\beta$ are simple sequences, and none of them breaks
every filesystem. If they contain the same command set, then (b) of
Theorem \ref{thm:struct} gives the claim, thus it suffices to show
that, e.g., each command in $\beta$ is also in $\alpha$. Let $\sigma\in\beta$ be
a command on node $n$, and $\FS$ be a filesystem on which both $\alpha$ and
$\beta$ work. If $\alpha$ does not contain any command on $n$, then
$(\beta\FS)(n)$ differs from $\FS(n) = (\alpha\FS)(n)$, contradicting
$\alpha\equiv\beta$, except when the replacement value in $\sigma$ is the
same as $\FS(n)$. In this case, however, $\FS(n)$ can be replaced by another
value from the same data type to force $\FS(n)$ and $(\sigma\FS)(n)$ be
different.
If $\alpha$ contains a command $\tau$ on $n$, then the
input types of $\sigma$ and $\tau$ must be the same, and the replacement
values of $\sigma$ and $\tau$ must be the same, thus $\sigma$ and $\tau$ are
the same commands.
\end{proof}

In the case when there is only one value in some data type, the second part
of the theorem does not remain true. For example, if there is only one
element in the empty data type, then $\(n,\E,\EE)$ commands are guaranteed not to make
any change to the filesystem, thus $\alpha$ and $\beta$ may contain
additional commands of this type without changing their semantics. If these
``no-operation'' commands are deleted from the equivalent $\alpha$ and
$\beta$ sequences extending their semantics, or if the rewriting rules are
extended with removal rules for such commands as in \cite{epcs}, then they become
rewritable.


\section{Update detection}\label{sec:update}

Given the original filesystem $\FS$ and its copy $\FS_1$ modified at
finitely many nodes, 
determine the simple update sequence $\alpha\in\Omega^*$ for which
$\alpha\FS=\FS_1$ as follows:
\begin{enumerate}
\item For each node $n\in\setn$ where $\FS(n)$ and $\FS_1(n)$ differ, add the
command $\(n,\tp(\FS(n)),\FS_1(n))$ to the command set $S$.
\item Order $S$ to become sequence $\alpha$ which honors $\orderrel$.
\end{enumerate}

The set $S$ can be created by traversing $\FS$ and $\FS_1$ simultaneously.
Every command set can be ordered to honor $\orderrel$ by first creating the
transitive closure of $\orderrel$, and then using topological sort. The
sorting procedure is clearly quadratic in the number of commands.

\begin{theorem}[Correctness of update detection]\label{thm:update-1}
The simple command sequence $\alpha\in\Omega^*$ returned by the update detector
works as expected: $\alpha\FS = \FS_1$.
\end{theorem}
\begin{proof}
By Theorem \ref{thm:sem-compl} there is a command sequence $\beta$
such that $\beta\FS=\FS_1$, and this command set consists of exactly the
commands in the above set $S$. By Theorem \ref{thm:struct} this set $S$ is
simple, and any ordering of $S$ honoring $\orderrel$ is 
semantically equivalent to $\beta$.
\end{proof}

\begin{theorem}\label{thm:update-2}
Suppose $\FS_1=\alpha^*\FS$ for some command sequence $\alpha^*$. Then there is a simple
sequence $\alpha$ such that $\FS_1=\alpha\FS$. This $\alpha$ can be computed
from $\alpha^*$ in quadratic time.
\end{theorem}

\begin{proof}
By Theorem \ref{thm:rewriting} there exists a simple sequence $\alpha$ such
that $\alpha^*\wxle\alpha$, and then $\alpha\FS=\FS_1$. The proof also
indicates a quadratic algorithm generating $\alpha$. For each command
$\sigma\in\alpha^*$ search backward from $\sigma$ to find the first command
$\tau\in\alpha^*$ which is on the same node as $\sigma$. If such a $\tau$ is
found, then use commutativity rules to move $\tau$ ahead and $\sigma$
backward until $\tau$ and $\sigma$ are next to each other, and then replace
them by a single command.
\end{proof}


\section{Refluent sequences}\label{sec:refluent}

Recall that two simple command sequences $\alpha$ and $\beta$ are
\emph{refluent} if there is a filesystem which neither $\alpha$ nor $\beta$
breaks. Using the $\models$ notation, it can be expressed as
$\{\alpha,\beta\}\nmodels\cbrk$.
This section starts with a characterization of refluent pairs in the special
case when the node sets of $\alpha$ and $\beta$ are disjoint. A general
reduction theorem together with a partial converse is provided for the case
when $\alpha$ and $\beta$ share commands on the same node.

\begin{theorem}[Refluent sequences]\label{thm:refluent}
The node-disjoint non-breaking simple sequences $\alpha$ and $\beta$ are
refluent if and only if
for each leader $\sigma$ in $\alpha$
one of the following conditions hold:
\begin{itemize}
\item $\sigma\indep\beta$, or
\item $\sigma$ matches $\(n,\E,\F\D)$ and there is a command on
$\parent n$ in $\beta$
matching $\(\parent n,\D,\F\E)$, or
\item $\sigma$ matches $\(\parent n,\D,\F\E)$ and some leader in
$\beta$ matches $\(n,\E,\F\D)$;
\end{itemize}
and the symmetric statements hold for each leader in $\beta$.
\end{theorem}

\begin{proof}
Assume first that $\alpha$, $\beta$ satisfy the above conditions. Create the
filesystem $\FS$ by repeating the process indicated in Theorem
\ref{thm:simple-completeness} for both $\alpha$ and $\beta$: start from the
empty filesystem, and for each leader in $\alpha$ and in $\beta$ execute the
described modifications of $\FS$.
Conditions of this theorem guarantee that Theorem \ref{non-breaking} applies
equally
to $\FS$ and $\alpha$ and to $\FS$ and $\beta$.

\smallskip

For the converse suppose $\alpha$ and $\beta$ are refluent and $\sigma$ is a
leader in $\alpha$. By Theorem \ref{thm:struct} there is an equivalent
permutation of $\alpha$ which starts with $\sigma$; and then the pair
$\{\sigma,\beta\}$ is refluent as well. Thus we may suppose that the first
sequence consists of this single leader only; let the node of $\sigma$ be
$n$.

If $\sigma\indep\beta$ then we are done, so suppose otherwise, which means
that there is a command $\tau\in\beta$ on a node $m$ comparable to $n$. 
Let $\FS$ be a filesystem neither $\sigma$ nor $\beta$ breaks.
If
$m$ is below $n$, then $\FS(n)$ is a directory (as $\tau$ is not an
$\(n,\E,\E)$ command and there are no commands on $n$ in $\beta$). As $\sigma$ changes $\FS(n)$ to a non-directory, every
node below $n$ is empty -- in particular, $\FS(m)$ is empty. Now $\tau$
changes $\FS(m)$ to a non-empty value, thus $\tau$ is an up command.
Consider the leader of the subtree of $\tau$, let it be $\tau'$ on node
$m'\preceq m$. There are no commands in $\beta$ above $m'$, and there are
structural commands in $\beta$ on every node between $m'$ and $m$.
Consequently $m'$ is below $n$, and just before $\tau'$ is executed, the
content of all nodes on and above $m'$ are the same as in $\FS$. As in $\FS$
all nodes below $n$ are empty, $m'$ must be the child of $n$, and the
leader $\tau'$ matches $\(m',\E,\F\D)$, while $\sigma$ matches $\(\parent
m',\D,\F\E)$.

The last case is when $n$ is below the node of some command in $\beta$.
Consider $\tau\in\beta$ on the node $m$ above $n$ such that no command in
$\beta$ is on a node between $n$ and $m$. As $\sigma$ is not a $\(n,\E,\E)$
command, and $\sigma\FS$ is not broken, all nodes above $m$ are directories.
$\tau$ is not $\(m,\D,\D)$, and when executed, the content at nodes between
$n$ and $m$ is the original value of $\FS$. As $\beta$ does not break $\FS$,
$m$ must be the parent of $n$, and $\FS(n)$ must be empty, leading to the
second possibility.
\end{proof}

Next we consider the general case when the simple sequences can share
commands on the same node. The following lemmas will be used later.

\begin{lemma}\label{lemma-4}
Suppose $\alpha$, $\beta$ are refluent simple sequences and $\tau_1$,
$\tau_2\in\beta$ are such that $\tau_1\orderrel\tau_2$. If $\tau_2\indep\alpha$,
then $\tau_1\indep\alpha$.
\end{lemma}
\begin{proof}
Consider first the case when $\tau_1\orderrel\tau_2$ matches
$\(n,\D\F,\E)\orderrel\(\parent n,\D,\F\E)$. It suffices to consider commands
$\sigma\in\alpha$ which are on a node comparable to $n$. If there is a $\sigma$ on
$n$, then its input type is the same as that of $\tau_1$, namely not $\E$,
and then $\tau_2\notindep\sigma$. For the same reason $\sigma$ cannot be on
$\parent n$. Thus $\sigma$ is either below $n$, or above $\parent n$, and in
both cases $\tau_2\indep\sigma$ implies $\tau_1\indep\sigma$.

In the second case the pair matches $\(\parent n,\E\F,\D)\orderrel \(n, \E,
\F\D)$. $\sigma\in\alpha$ cannot be on $\parent n$ as the input type of
$\tau_1$ is not $\D$, and so $\tau_2\notindep\sigma$ would hold. Otherwise,
if $\sigma$ is above
$\parent n$ then $\tau_2\indep\sigma$ implies $\tau_1\indep\sigma$. In the
remaining cases the node $m$ of $\sigma$ is below $\parent n$. The input
type of $\tau_1$ is not a directory, thus every command below $\parent n$
must have $\E$ as input type by Theorem \ref{non-breaking}. If the output
type of $\sigma$ is also $\E$, then we are done. If not, then $\sigma$ is an
up command, and consider the leader (in $\alpha$) of $\sigma$; there is a
command in $\alpha$ on every node between $m$ and the leader. Every node
above the leader is a directory (Theorem \ref{non-breaking}), thus this
leader must be on or above $\parent n$. But then $\alpha$ has a command on
$\parent n$, which is a contradiction.
\end{proof}

\begin{lemma}\label{lemma-2}
Suppose $\tau\alpha$ and $\beta$ are non-breaking simple sequences, and
$\tau\indep\beta$. If $\alpha$ and $\beta$ are refluent, then so are
$\tau\alpha$ and $\beta$.
\end{lemma}
\begin{proof}
Let $\tau$ be on the node $n$, and $\FS$ be a filesystem which neither
$\alpha$ nor $\beta$ breaks. Our aim is to construct a filesystem $\FSa$ on
which both $\tau\alpha$ and $\beta$ work. As $\tau\alpha$ is non-breaking,
by Proposition \ref{claim:add-one} either $\tau\indep\alpha$ or
$\tau\orderrel\sigma$ for some $\sigma\in\alpha$. 

If $\tau\indep\alpha$, then to get $\FSa$,
change $\FS(n)$ to a value matching the input type of $\tau$, all nodes
above $n$ to a directory value (except when $\tau$ matches $\(n,\E,\E)$), and
all nodes below $n$ to an empty value (except when $\tau$ matches
$\(n,\D,\D)$). Clearly, $\tau$ is applicable to $\FSa$, and according to
Corollary \ref{file-comparable}, neither $\beta$ breaks $\FSa$, nor $\alpha$
breaks $\tau\FSa$.

If $\tau\orderrel\sigma$ matches $\(n,\D\F,\E)\orderrel\(\parent n,\D,\F\E)$,
then $\FS(\parent n)$ is a directory, and $\FS(n)$ is empty by Theorem 
\ref{non-breaking} (as $n$ is not mentioned in $\alpha$). To get
$\FSa$, change $\FS$ only at $n$ to a value matching the input type of
$\tau$. Then $\FSa$ is a filesystem, and due to Corollary \ref{file-comparable}
$\beta$ does not break $\FSa$, and $\alpha$ does not break
$\tau\FSa$.

Finally, if $\tau$ matches $\(n,\E\F,\D)$, then $\sigma$ is an up command on
a child of $n$, consequently $\FS(n)$ is a directory. All commands in
$\alpha$ and in $\beta$ with a node below $n$ have the empty input type. In this
case in $\FS$ change the content of every node below $n$ to an empty value,
and at $n$ to a value matching the input type of $\tau$. Then $\FSa$ is a
filesystem, and, as before, both $\tau\alpha$ and $\beta$ can be applied to
$\FSa$.
\end{proof}

Recall that $\tp(\sigma)=\tp(\tau)$ if these commands are on the same node,
and have the same input and the same output types.

\begin{lemma}\label{lemma-3}
Suppose $\alpha$ and $\beta$ are refluent simple sequences,
$\sigma\in\alpha$, $\tau\in\beta$, and $\sigma\orderrel\tau$. Then
there is a $\sigma'\in\beta$ such that $\tp(\sigma)=\tp(\sigma')$.
\end{lemma}

\begin{proof}
Let $\FS$ be a filesystem on which both $\alpha$ and $\beta$ work, and
suppose by contradiction that $\beta$ has no structural command on the node
of $\sigma$. Distinguish two cases based on the pattern of $\sigma\orderrel
\tau$. If it is $\(n,\D\F,\E) \orderrel \(\parent n, \D,\F\E)$, then $\FS(n)$
matches the input type of $\sigma$ (as otherwise $\alpha$ would break
$\FS$), in particular $\FS(n)$ is not empty. As $\beta$ has no structural
command on $n$, $\beta$ does not change the type of $\FS(n)$. But when
$\tau$ is executed, $\FS(\parent n)$ becomes a non-directory, which breaks
the filesystem.

In the second case the pattern is $\(\parent n,\E\F,\D) \orderrel
\(n,\E,\F\D)$. Then $\FS(\parent n)$ is not a directory and $\beta$ keeps the
value type stored here. However $\tau$ adds a non-empty value at $n$
breaking the filesystem.

Thus $\beta$ contains a structural command $\sigma'$ on the node of
$\sigma$. Since $\alpha$ and $\beta$ are refluent simple sequences, the
input types of $\sigma$ and $\sigma'$ are the same (both are
applicable to the same node). By Theorem
\ref{thm:struct} $\sigma'\orderrel\tau$ must also hold (as $\sigma'$ and
$\tau$ are structural commands on immediate relatives in a non-breaking
simple sequence). Combined with
the fact that $\sigma\orderrel\tau$ and $\sigma'\orderrel\tau$ implies that
$\sigma$ and $\sigma'$ have the same output type, the lemma follows.
\end{proof}

\begin{definition}
Write $\captp\alpha\beta$ for the set of commands from $\alpha$ which are on
the same node, have the same input type and the same output type as some
command in $\beta$:
$$
   \captp\alpha \beta = \{\sigma\in\alpha: \tp(\sigma)=\tp(\tau)
       \mbox{ for some } \tau\in\beta\,\}.
$$
\end{definition}
Clearly, if $\alpha$ and $\beta$ are simple, the elements of
$\captp\alpha\beta$ and $\captp\beta\alpha$ are in a
one-to-one correspondence: a command in one set corresponds to the command
in the other set on the same node; the pairs share the node, the input type,
and the output type.

\begin{theorem}[Reduction of refluent sequences]\label{thm:reduction}
Suppose $\alpha$ and $\beta$ are refluent simple sequences. Then
$\alpha\equiv\alpha_1\alpha_2$ and $\beta\equiv\beta_1\beta_2$ where
$\alpha_1$ consists of commands in $\captp\alpha\beta$, and $\beta_1$
consists of commands in $\captp\beta\alpha$. Furthermore $\alpha_2$ and
$\beta_2$ are refluent.
\end{theorem}

\begin{proof}
The first part of the theorem follows from the structural Theorem
\ref{thm:struct} after we show that if for any two commands
$\sigma \orderrel\tau$ from $\alpha$ such that $\tau\in\alpha_1$, then
$\sigma\in\alpha_1$ as well. But this is immediate from Lemma
\ref{lemma-3}.

To see the second part, observe that the simple command sets of $\alpha_1$
and $\beta_1$ are on the same node set with the same input and output types.
Consequently for any filesystem $\FS$, $\alpha_1\FS$ and $\beta_1\FS$ have
the same data type (but not necessarily the same value) at every node. Thus
if $\alpha_2(\alpha_1\FS)\neq\bot$ and $\beta_2(\beta_1\FS)\neq\bot$, then
$\beta_2(\alpha_1\FS)\neq\bot$ as well, showing that $\alpha_2$ and
$\beta_2$ are refluent indeed.
\end{proof}

A partial converse of Theorem \ref{thm:reduction} is true.

\begin{theorem}\label{thm:converse}
Suppose $\gamma\alpha$ and $\gamma\beta$ are non-breaking simple sequences.
They are refluent if and only if $\alpha$ and $\beta$ are refluent.
\end{theorem}

\begin{proof}
The direction that if $\gamma\alpha$ and $\gamma\beta$ are refluent, then so
are $\alpha$ and $\beta$ is clear. The other direction follows from the
special case when $\gamma$ consists of a single command $\tau$. For this
case, however, an easy adaptation of the proof of Lemma \ref{lemma-2} works.
\end{proof}


\section{Reconciliation}\label{sec:reconciliation}

Let us revisit the problem of file synchronization. We have two simple
sequences $\alpha$ and $\beta$ which create two divergent replicas of the same original
filesystem. The goal is to find (preferably maximal) subsets which can
then be carried over to the other copy without destroying local
modifications.

\begin{definition}
(a) The sequence $\beta'$ formed from commands in $\beta$ is a
\emph{reconciler for $\alpha$ over $\beta$}, if $\beta'$ does not destroy
any local change made by $\alpha$, and is always applicable after $\alpha$
in the sense that $\{\alpha,\beta \} \models \alpha\beta'$.

\noindent (b)
The sequences $\alpha$ and $\beta$ are \emph{confluent} if they are
refluent, and there are reconcilers $\beta'$ and $\alpha'$
which create identical results,
written succinctly as $\{\alpha,\beta\}\models
\alpha\beta'\equiv\beta\alpha'$.
\end{definition}

\begin{theorem}[Confluent node-disjoint sequences]\label{thm:confluent}
The node-disjoint non-break\-ing simple sequences $\alpha$ and $\beta$ are
confluent if and only if $\alpha\indep\beta$. In this case the
reconciler sequences are $\beta$ and $\alpha$ respectively as
$\{\alpha,\beta\}\models\alpha\beta\equiv\beta\alpha$.
\end{theorem}

\begin{proof}
Suppose $\alpha$ and $\beta$ are applied to the filesystem $\FS$. We may
assume that each command actually changes the content of the filesystem. Let
$\beta'$ and $\alpha'$ be the reconcilers, that is, $\alpha\beta'
\equiv \beta\alpha'$. If $\sigma\in\beta$ on node $n$ were not in $\beta'$,
then $(\alpha\beta')\FS$ has the original content at $n$, while
$(\beta\alpha')\FS$ has a different value as changed by $\sigma$. Thus
$\beta'$ contains all commands of $\beta$, and as it is non-breaking, it is
equivalent to $\beta$ by Theorem \ref{thm:struct}. Consequently
$\beta'=\beta$ and $\alpha'=\alpha$ satisfy $\alpha\beta\equiv \beta\alpha$.
As $\alpha\beta$ and $\beta\alpha$ are non-breaking simple sequences, both
honor $\orderrel$. By Theorem \ref{thm:struct} it means that
$\alpha\indep\beta$.

For the other direction assume $\alpha$, $\beta$ are node-disjoint,
non-breaking sequences such that $\alpha\indep\beta$. By Theorem
\ref{thm:struct} $\alpha\beta\equiv\beta\alpha$, thus it suffices to show
that $\{\alpha,\beta\}\models \alpha\beta$.
The command sets of $\alpha$ and $\beta$ are simple. As $\alpha\indep\beta$,
the same is true for the command set of $\alpha\cup\beta$; moreover
$\alpha\beta$ is an ordering of this simple set which honors $\orderrel$. In
particular, $\alpha\beta$ is non-breaking. Let $\FS$ be a filesystem which
neither $\alpha$ nor $\beta$ breaks, that is, conditions of Theorem
\ref{non-breaking} hold for $\alpha$ and $\beta$. Since $\alpha\indep
\beta$, every leader in $\alpha\beta$ is either a leader in $\alpha$, or is
a leader in $\beta$. From here it follows that the same conditions hold for
the sequence $\alpha\beta$, meaning $(\alpha\beta)\FS\neq\bot$, as required.
\end{proof}

To state the main result of this section we need some additional
definitions. Recall that $\captp\alpha\beta$ is the set of those commands
from $\alpha$ which have the same node, same input type and same output type
as some command in $\beta$.

\begin{definition}
The set of commands in $\alpha$ not in $\captp\alpha\beta$ is denoted by
$\mintp\alpha\beta$ as
$$
   \mintp\alpha\beta = \{ \sigma\in\alpha:  \mbox{for every $\tau\in\beta$, } 
         ~\tp(\sigma)\neq\tp(\tau)\,\}.
$$
Finally, let us define
$$
   \R\beta\alpha = 
     \{ \tau\in\mintp\beta\alpha : \tau\indep \mintp\alpha\beta \,\}.
$$
When this set is used as a sequence, it is ordered so that the ordering
honors $\orderrel$. Any two such ordering gives equivalent sequences by
Theorem \ref{thm:struct}.
\end{definition}

\begin{theorem}\label{thm:reconc}
Let $\alpha$, $\beta$ be refluent simple sequences. Then\\
{\upshape(a)} $\R\beta\alpha$ is a reconciler for $\alpha$ over 
$\beta$.\\
{\upshape(b)} If $\beta'$ is a reconciler, then
$\beta'\subseteq \R\beta\alpha$.
\end{theorem}

\begin{proof}
(a) By the Reduction Theorem \ref{thm:reduction}, $\alpha$ and $\beta$ can
be equivalently rearranged as $\alpha_1\alpha_2$ and $\beta_1\beta_2$ where
$\alpha_1$ consists of the commands of $\captp\alpha\beta$, $\alpha_2$
consists of the commands of $\mintp\alpha\beta$, and similarly for $\beta_1$
and $\beta_2$. Recall that by the same theorem $\alpha_2$ and $\beta_2$ are
also refluent.

We claim that $\beta_2$ can be rearranged so that it starts with
$\R\beta\alpha$. To this end we only need to show that if $\tau_1\orderrel
\tau_2$ are in $\beta_2$ and $\tau_2\in\R\beta\alpha$, then so is $\tau_1$.
By definition $\tau\in\beta_2$ is in $\R\beta\alpha$ iff
$\tau\indep\alpha_2$. Thus $\tau_2\indep\alpha_2$, and then by Lemma
\ref{lemma-4} we have $\tau_1\indep\alpha_2$, as required.

Therefore $\beta_2\equiv\beta'\beta''$ where $\beta'$ consists of the
commands in $\R\beta\alpha$. Now $\beta'$ and $\alpha_2$ are node-disjoint
non-breaking refluent sequences such that $\beta'\indep\alpha_2$. Theorem
\ref{thm:confluent} gives that in this case $\{\alpha_2,\beta'\} \models
\alpha_2\beta'$, and then by Theorem \ref{thm:converse} we have
$\{\alpha,\beta\}\models\alpha\beta'$, proving that $\R\beta\alpha$ is
indeed a reconciler.

\smallskip

(b) Suppose $\beta'$ is a reconciler, in particular $\{\alpha,
\beta\}\models\alpha\beta'$. Then $\alpha\beta'$ and $\beta$ are also
refluent, consequently, by Theorem \ref{thm:reduction}, $\alpha\equiv \alpha_1 \beta'
\alpha_2$ and $\beta\equiv\beta_1\beta'\beta_2$ where $\alpha_1$ and
$\beta_1$ are the commands from $\captp\alpha\beta$ and $\captp\beta\alpha$,
respectively. As $\beta'\alpha_2$ is a non-breaking simple sequence, if
$\beta'\notindep\alpha_2$ then according to Theorem \ref{thm:struct} there
are $\tau\in\beta'$ and $\sigma\in\alpha_2$ such that $\tau\orderrel\sigma$.
But $\alpha$ and $\beta$ are refluent, $\tau\in\beta$, $\sigma\in\alpha$,
and then Lemma \ref{lemma-3} gives that there is a $\tau'\in\alpha$ such
that $\tp(\tau')=\tp(\tau)$ meaning that $\tau\in\beta_1$,
which is impossible. Thus $\beta'\indep\alpha_2$ and
then $\beta'\subseteq\R\beta\alpha$.
\end{proof}
 
It should be clear that $\R\beta\alpha$ can be determined from the simple 
sequences $\alpha$ and $\beta$ in quadratic time. Split the commands in
$\alpha$ into the sets $\captp\alpha\beta$ and $\mintp\alpha\beta$, and
similarly for $\beta$. Then go over each element of $\mintp\beta\alpha$ and check
whether it satisfies $\tau\indep\mintp\alpha\beta$. As $\beta$ honors
$\orderrel$, keeping elements of $\R\beta\alpha$ in the same order as they
are in $\beta$ provides a correct ordering of $\R\beta\alpha$.

\smallskip

We can write the refluent sequences as $\alpha \equiv \alpha_1\R\alpha\beta
\alpha_3$ where $\alpha_1$ consists of commands in $\captp\alpha\beta$, and
similarly for $\beta$. Unresolved conflicts come from two sources. First,
the matching commands in $\alpha_1$ and $\beta_1$ might store different
values (of the same type) at the same node, which would override a local
change made by $\alpha$. These conflicts should be
resolved by some content negotiation. Second, a command $\sigma\in \beta_3$
is either on the same node as some command in $\alpha$ (actually, in
$\alpha_3$) assigning a different value type thus again overriding a local
change, or $\sigma\notindep\alpha_3$. In this latter case
executing $\sigma$ after $\alpha$ (or even after $\alpha\R\beta\alpha$)
would break the filesystem.

\smallskip

Finally, let us state an immediate consequence of Theorem \ref{thm:reconc}
of which Theorem \ref{thm:confluent} is a special case.

\begin{theorem}\label{thm:conf-gen}
The refluent simple sequences $\alpha$ and $\beta$ are confluent if and only
if $\captp\alpha\beta = \captp\beta\alpha$ and $\mintp\alpha\beta \indep
\mintp\beta\alpha$.
\qed
\end{theorem}


\section{Conclusion}\label{sec:conclusion}

Our main motivation was to generalize and extend some of the results from
\cite{epcs} in a more abstract setting, and concentrated mainly on proving
several characterization results of this intriguing algebraic model. One of
the main contributions of \cite{epcs} is the model of filesystems and
filesystem commands, which have been adopted here as well. The semantics of
command sequences is defined through their action on the filesystems, which
gives rise to semantic equivalence and semantic validity.
Two sequences are semantically equivalent if they have the same effect on
all filesystems, while the sequence $\alpha$ is semantically valid on a
filesystem $\FS$ if $\alpha$ can be executed on $\FS$ without breaking it.
This semantical validity shares many properties of the ``logically valid''
notion of mathematical logic as discussed in Proposition \ref{prop:models}.
Similarly, the \emph{non-breaking} property of sequences corresponds to that
of consistency or satisfiability in logic. The set of filesystem commands is
also \emph{functionally complete}: if two filesystems differ at finitely
many nodes, then there is a \emph{simple} command sequence transforming one
into the other as shown by Theorem \ref{thm:sem-compl}.

Command sequences can be manipulated syntactically by applying the
\emph{rewriting rules} defined in Proposition \ref{prop:rules}. Every
sequence can be rewritten into a simple sequence while extending its semantics
(Theorem \ref{thm:rewriting}), and two simple sequences are semantically
equivalent if and only if they can be rewritten into each other (Theorem
\ref{thm:simple-completeness}). By Theorem \ref{thm:struct} the semantics of
a simple sequence is uniquely determined by the \emph{set of its commands};
each feasible ordering (which can be found in quadratic time) of such a
simple set gives a sequence with the same semantics. The existence of an
effective update detector algorithm follows easily from these properties
yielding the statements in Theorem \ref{thm:informal-1}.

Two simple sequences are \emph{refluent} if they are jointly consistent:
there is a filesystem which neither of them breaks. The problem of
syntactical characterization of such sequence pairs seems to be hard and
have been solved only partially. Theorem \ref{thm:refluent} gives a complete
characterization for the special case of node-disjoint sequences.
By Theorems \ref{thm:reduction} and \ref{thm:converse} commands on the same
node with the same input and same output type can be ignored.

Reconciliation is a relaxed notion of confluence: given two simple sequences
$\alpha$ and $\beta$, can we add further commands (without overriding the
effects
of old ones) to these sequences so that $\alpha\beta'\equiv\beta \alpha'$ ?
Theorem \ref{thm:conf-gen} gives a complete characterization when sequences
are confluent, while by Theorem \ref{thm:reconc} there is a unique maximal 
reconciler for any pair of refluent simple sequences. This result justifies
Theorem \ref{thm:informal-2}.
We have assumed that all three data types $\D$, $\F$, $\E$ contain at least
two elements; this fact was used in the proof of the Completeness theorem
\ref{thm:simple-completeness}. When one (or more) of them has a single
element only, the corresponding transient command should be deleted, as
explained at the end of Section \ref{sec:simple-sequences}. With that
modification all theorems remain valid.


\subsection{Open problems and future work}

During the preparation of this work we have looked at changing, relaxing, or
modifying several parameters of the chosen model, without any success. One
such modification was to use more value types beyond $\D$, $\F$, and $\E$. 
The resulting
filesystem semantics (using various restrictions on how the values vary
along each branch) with the corresponding command set did not lead to any
syntactical characterization of semantically equivalent sequences. It is an
interesting and intriguing problem to understand why this particular
semantics is so powerful and, at the same time, so tractable.

It is an open problem whether a reconciler with similar properties as claimed in 
Theorem \ref{thm:informal-2} exists for more than two replicas. While many
synchronization algorithms generalize easily to many replicas, it is not
clear how the constructions from Section \ref{sec:refluent} can be
extended.

The reconciler algorithm extracts the reconciler sequences to be applied to
the replicas, and marks the remaining commands as conflicting. The
downstream conflict resolver may resolve these conflicts by selecting the
leaders from the conflicting commands, consider the pairs that caused the
conflict, and decide to roll one of them back. Further research could
establish how these decisions affect the outcome; what rollbacks and what
order of them would undo the minimum amount of changes, and whether there is
a theoretically optimal or otherwise preferred set of rollbacks to perform.
Also, as noted earlier, extensions of the current model could be introduced
to handle links and a convenient \emph{rename} operation.

Another possible extension could be to include additional \emph{edit} commands that are
commutative like summation, decrementing or incrementing, which are common
operations in various database applications. Extending the algebra to
consider these commands can reduce the number of conflicts and increase the
power and efficiency of the reconciler in itself.

We also hope that this work provides a blueprint of constructing an algebra
of commands for different storage protocols (e.g., mailbox folders, generic
relational databases, etc.), and of demonstrating the adequacy and
completeness of the update and conflict detection and reconciliation
algorithms defined over it. This, in turn, can offer formal verification of
the algorithms underlying specific implementations in a variety of
synchronizers.


\section*{Acknowledgment}
The research of the second author (L.~Cs.) has been supported by the
GACR project number 19-04579S and by the Lend\"ulet Program of the HAS,
which is thankfully acknowledged.



\end{document}